\documentclass{article}

\usepackage{amssymb}
\usepackage{amsmath}
\usepackage{amsthm}
\usepackage[dvips]{graphicx}

\newtheorem{thm}{Theorem}
\newtheorem{lem}{Lemma}

\newtheorem{de}{Definition}
\newtheorem{rem}{Remarks}
\newtheorem{prop}{Proposition}
\newtheorem{cor}{Corollary}
\newtheorem{exam}{Example}

\begin{document}

\title{Reaction Automata}

\author{{\bf Fumiya Okubo$^1$}, {\bf Satoshi Kobayashi$^2$}  and {\bf Takashi Yokomori$^3$\footnote{Corresponding author}}\\[2mm]
$^1$Graduate School of Education\\
Waseda University, 1-6-1 Nishiwaseda, Shinjuku-ku\\
Tokyo 169-8050, Japan\\
{\tt f.okubo@akane.waseda.jp}\\[2mm]
$^2$Graduate School of Informatics and Engineering,\\
 University of Electro-Communications, \\
 1-5-1 Chofugaoka, Chofu-shi, Tokyo 182-8585, Japan\\
{\tt satoshi@cs.uec.ac.jp}\\[2mm]
$^3$Department of Mathematics\\
 Faculty of Education and Integrated Arts and Sciences\\
 Waseda University, 1-6-1 Nishiwaseda, Shinjuku-ku\\
Tokyo 169-8050, Japan\\
{\tt yokomori@waseda.jp}}

\date{}
\maketitle

\begin{abstract}
Reaction systems are a formal model that has been introduced to investigate  the interactive behaviors of biochemical reactions. Based on the formal framework of reaction systems, we propose new computing models  
called {\it reaction automata} that feature (string) language acceptors with multiset manipulation as a computing mechanism, and show that 
reaction automata are computationally Turing universal. Further, some subclasses of reaction automata with space complexity are investigated and their language classes are compared to the ones in the Chomsky hierarchy.  
\end{abstract}

\section{Introduction}

In recent years, a series of seminal papers \cite{ER:07a,ER:07b,ER:09} has been published in which  Ehrenfeucht and Rozenberg have introduced a formal model, called {\it reaction systems}, for investigating interactions between biochemical reactions, where  two basic components (reactants and inhibitors) are employed as regulation mechanisms for controlling biochemical functionalities.  It has been shown that reaction systems provide a formal framework best suited for investigating in an abstract level the way of emergence and evolution of biochemical functioning such as events and modules. In the same framework, they also introduced the notion of time into reaction systems and investigated notions such as reaction times, creation times of compounds and so forth.  Rather recent two papers  \cite{EMR:10,EMR:11} continue the investigation of reaction systems, with the focuses on combinatorial properties of functions defined by random reaction systems and on the dependency relation between the power of defining functions and  the amount of available resource.  

In the theory of reaction systems, a (biochemical) reaction is formulated as a triple $a=(R_a, I_a, P_a)$, where $R_a$ is the set of molecules called {\it reactants}, $I_a$ is the set of molecules called {\it inhibitors}, and $P_a$ is the set of molecules called {\it products}. Let $T$ be a set of molecules, and the result of applying a 
 reaction $a$ to $T$, denoted by $res_a(T)$, is given by $P_a$ 
if $a$ is enabled by $T$ (i.e., if $T$ completely includes $R_a$ 
and excludes $I_a$). Otherwise, the result is empty.  Thus, $res_a(T)=P_a$ if $a$ is enabled on $T$, and $res_a(T)=\emptyset$ otherwise. The result of applying a reaction $a$ is extended to the set of reactions $A$, denoted by $res_A(T)$, and an interactive process 
consisting of a sequence of $res_A(T)$'s is properly introduced and investigated. 

In the last few decades, the notion of a multiset has frequently 
appeared and been investigated in many different areas such as 
mathematics, computer science, linguistics, and so forth. 
(See, e.g., \cite{CPRS:01} for the reference papers written from the 
viewpoint of mathematics and computer science.) The notion of a multiset has received more and more attention, particularly in the areas of biochemical computing and molecular computing (e.g., \cite{Paun:00,Set:01}).  

Motivated by these two notions of a reaction system and a multiset, in this paper we will introduce  computing devices called {\it reaction automata} and  show  that they are computationally universal by proving that any recursively enumerable language is accepted by a reaction automaton.  
There are two points to be remarked:  On one hand, the notion of reaction automata may be taken as  a kind of an extension of reaction systems in the sense that our reaction automata deal with {\it multisets} rather than (usual) sets  as reaction systems do, in the sequence of computational process.  On the other hand, however,       
reaction automata are introduced as computing devices that accept 
the sets of {\it string objects} (i.e., languages over an alphabet). 
This unique feature, i.e., a string accepting device based on multiset 
computing in the biochemical reaction model can be realized by introducing a simple idea of feeding an input to the device from the environment and by employing a special encoding technique.  

\begin{figure}[t]
\centerline{
\includegraphics[scale=0.2]{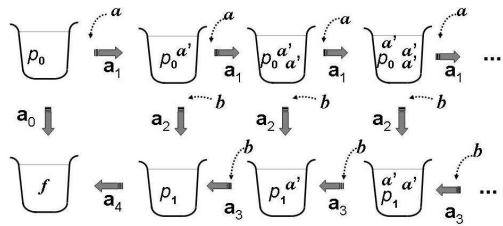}}
\caption{A graphic illustration of interactive biochemical reaction processes for accepting  strings in the language $L=\{a^nb^n \mid n\geq 0\}$ in terms of our reaction automaton  $\mathcal{A}$. }
\label{graphic}
\end{figure}

In order to illustrate an intuitive idea of the notion of reaction automata and their behavior, we give in Figure \ref{graphic} a simple example of the behavior of a reaction automata $\mathcal {A}$ that consists of the set of objects $\{p_0, p_1, a, b, a', f \}$ (with the input alphabet $\{a, b\}$), 
 the set of reactions 
$\{{\bf a}_0=(p_0, \{a,b,a'\}, f), {\bf a}_1=(p_0a,\{b\}, p_0a'), {\bf a}_2=(p_0a'b,\emptyset, p_1), {\bf a}_3=(p_1a'b,\{a\}, p_1), 
{\bf a}_4=(p_1,\{a,b,a'\}, f)\}$, 
where $\{p_0\}$ is the initial multiset and $\{f\}$ is the final multiset. 
Note that in a reaction ${\bf a}=(R_a, I_a, P_a)$, multisets $R_a$ and $P_a$ are  represented by string forms, while $I_a$ is given as a set.  
In the graphic drawing of Figure \ref{graphic},  each reaction ${\bf a}_i$ 
is applied to a multiset (of a test tube) after receiving an input symbol (if any is provided) from the environment. In particular,  applying ${\bf a}_0$ to $\{p_0\}$ leads to that the empty string is  accepted by $\mathcal{A}$. It is seen, for example, that reactions ${\bf a}_1$ and ${\bf a}_2$ are enabled by the multiset $T=\{p_0,a',a'\}$ only when inputs 
$a$ and $b$, respectively, are received,  which result in producing $R_1=\{p_0,a',a',a'\}$ 
and $R_2=\{p_1,a'\}$, respectively. Thus, we have that $res_{{\bf a}_1}(T\cup \{a\})=R_1$ and $res_{{\bf a}_2}(T\cup \{b\})=R_2$.  
Once applying ${\bf a}_2$ has brought about a change of $p_0$ into $p_1$, $\mathcal{A}$ has no possibility of accepting further inputs $a$'s, because of the inhibitors in ${\bf a}_3$ or ${\bf a}_4$.  One may easily see that $\mathcal{A}$ accepts the language $L=\{a^nb^n\mid n\geq 0\}$. We remark that     reaction automata allow a multiset of reactions $\alpha$ to apply to a multiset of objects $T$ in an exhaustive manner (what we call {\it maximally parallel manner}),  and therefore the interactive process sequence of computation is nondeterministic in that the reaction result from  $T$ may produce more than one product. The details for these are formally described in the sequel.

This paper is organized as follows.  After preparing the basic notions and notations from formal language theory in Section 2, we formally introduce  the main notion of reaction automata together with one language example in Section 3. Then, Section 4 describes a multistack machine (in fact, two-stack machine) whose specific property will be demonstrated to be very useful in the proof of the main result in the next section. Thus, in Section 5 we present our main results: reaction automata  are 
computationally universal. We also consider some subclasses of reaction automata from a viewpoint of the complexity theory in Section 6, and 
investigate the language classes accepted by those subclasses in comparison to the Chomsky hierarchy. 
Finally,  concluding remarks as well as future research topics are briefly discussed in Section 7.

\section{Preliminaries}
We assume that the reader is familiar with the basic notions of formal language theory.  For unexplained details, refer  to~\cite{HMU:03}. 

Let $V$ be a finite alphabet. For a set $U \subseteq V$, the cardinality of $U$ is denoted by $|U|$. The set of all finite-length strings over $V$ is denoted by $V^*$. The empty string is denoted by $\lambda$.  For a string $x$ in $V^*$,  $|x|$ denotes the length of $x$, while for a symbol $a$ in $V$ we denote  by  $|x|_a$ the number of occurences of $a$ in $x$.  
For $k \ge 0$, let $pref_k (x)$ be the prefix of a string $x$ of length $k$. 
For a string $w=a_1a_2\cdots a_n \in V^*$, $w^R$ is the reversal of $w$, that is, 
$(a_1 a_2 \cdots a_n)^R = a_n \cdots a_2 a_1$.
Further, for a string $x=a_1 a_2 \cdots a_n \in V^*$,  $\hat{x}$ denotes the hat version of $x$, i.e.,  
$\hat{x} = \hat{a_1} \hat{a_2} \cdots \hat{a_n}$, where each 
 $\hat{a}_i$ is in an alphabet $\hat{V} = \{ \hat{a} \, | \, a \in V \}$ such that $V \cap \hat{V} = \emptyset$. 

We use the basic notations and definitions regarding multisets that follow~\cite{CMM:01,KMP:01}.
A {\it multiset} over an alphabet $V$ is a mapping $\mu: V \rightarrow \mathbf{N}$, 
where $\mathbf{N}$ is the set of non-negative integers and for each $a \in V$, $\mu(a)$ represents the number of occurrences of $a$ in the multiset $\mu$.  
The set of all multisets over $V$ is denoted by $V^\#$, including the empty multiset denoted by $\mu_{\lambda}$, where $\mu_{\lambda}(a)=0$ for 
all $a\in V$.  A multiset $\mu$ may be represented as a vector, $\mu(V) = (\mu(a_1), \dots, \mu(a_n))$, for an ordered set $V = \{ a_1, \dots, a_n \}$.  We can also represent the multiset $\mu$ by any permutation of the string $w_\mu = a^{\mu(a_1)}_1 \cdots a^{\mu(a_n)}_n$. Conversely, with any string $x \in V^*$ one can associate the multiset $\mu_x : V \rightarrow \mathbf{N}$ defined by $\mu_x(a) = |x|_a$ for each $a \in V$. In this sense, we often identify a multiset $\mu$ with its string representation $w_{\mu}$ or any permutation of $w_{\mu}$.  Note that the string representation of $\mu_{\lambda}$ is $\lambda$, i.e., $w_{\mu_{\lambda}}=\lambda$.  

A usual set $U \subseteq V$ is regarded as a multiset $\mu_U$ such that 
 $\mu_U(a) = 1$ if  $a$ is in $U$ and $\mu_U(a)=0$ otherwise.  In particular, 
 for each symbol $a \in V$, a multiset $\mu_{\{a\}}$ is often denoted by $a$ itself.

For two multisets $\mu_1$, $\mu_2$ over $V$,  we define one relation and three operations as follows: 
\[
\begin{array}{ll}
{\it Inclusion}: &\mu_1 \subseteq \mu_2 \text{ \  iff \ } \mu_1(a) \le \mu_2(a),  \text{ for each }  a \in V, \\
{\it Sum}:&(\mu_1 + \mu_2) (a) = \mu_1(a) + \mu_2(a),  \text{ for each }  a \in V,\\
{\it Intersection}:&(\mu_1 \cap \mu_2) (a) = {\rm min}\{\mu_1(a), \mu_2(a)\},  \text{ for each } a \in V,\\
{\it Difference}:&(\mu_1 - \mu_2) (a) = \mu_1(a) - \mu_2(a),  \text{ for each }  a \in V \text{ (for the case } \mu_2 \subseteq \mu_1 \text{ only)}.
\end{array}
\]
A multiset $\mu_1$ is called {\it multisubset} of $\mu_2$ if $\mu_1 \subseteq \mu_2$.  
The sum for a family of multisets $\mathcal{M} = \{\mu_i \}_{i \in I}$ is also denoted by $\sum_{i \in I}\mu_i$. For a multiset $\mu$ and $n \in \mathbf{N}$, $\mu^n$ is defined by $\mu^n (a) = n \cdot \mu(a)$ for each $a \in V$. The {\it weight} of a multiset $\mu$ is $|\mu| = \sum_{a \in V} \mu(a)$.

We introduce an injective function $stm : V^* \to V^\#$ that maps a string to a multiset in the following manner: 
\begin{eqnarray*}
\left\{ \begin{array}{ll}
 stm(a_1 a_2 \cdots a_n) = a_1 a_2^2 \cdots a_n^{2^{n-1}}  & \mbox{(for $n\geq 1$)} \\
stm(\lambda) = {\lambda}. &  \\
\end{array} \right.
\end{eqnarray*}

\section{Reaction Automata}

As is previously mentioned, a novel formal model called reaction systems has been introduced in order to investigate the property of interactions between biochemical reactions, where two basic components (reactants and inhibitors) are employed as regulation mechanisms 
for controlling biochemical functionalities (\cite{ER:07a,ER:07b,ER:09}). 
Reaction systems provide a formal framework best suited for investigating the way of emergence and evolution of biochemical functioning on an abstract level.  

By recalling  from \cite{ER:07a} basic notions related to reactions systems, we first 
extend them (defined on the sets) to the notions on the multisets.  Then, we shall  
introduce our notion of {\it reaction automata} which plays  a central role in this paper.

\begin{de}
{\rm 
For a set $S$,  a {\it reaction} in $S$ is a 3-tuple ${\bf a} = (R_{\bf a}, I_{\bf a}, P_{\bf a})$ of finite multisets, such that $R_{\bf a}, P_{\bf a} \in S^\#$, $I_{\bf a} \subseteq S$ and $R_{\bf a} \cap I_{\bf a} = \emptyset$.
}
\end{de}
The multisets $R_{\bf a}$ and $P_{\bf a}$ are called the {\it reactant} of ${\bf a}$ and the {\it product} of ${\bf a}$, respectively, while the set $I_{\bf a}$ is called the {\it inhibitor} of ${\bf a}$. These notations are extended to a multiset of reactions as follows:    For a set of reactions $A$ and a multiset $\alpha$ over $A$,  
\[ R_{\alpha} = \sum_{{\bf a}\in A} R_{\bf a}^{\alpha({\bf a})}, \, \, I_{\alpha} = \bigcup_{ {\bf a} \subseteq \alpha} I_{\bf a}, \, P_{\alpha} = \sum_{{\bf a}\in A} P_{\bf a}^{\alpha({\bf a})}. \] 

In what follows,  we usually identify the set of reactions $A$ with the set of labels $Lab(A)$ of reactions in $A$, and often use the symbol $A$ as a finite alphabet.

\begin{de}
{\rm 
Let $A$ be a set of reactions in $S$ and  $\alpha \in A^\#$ be a multiset of reactions 
over $A$.  Then, for  a finite multiset $T \in S^\#$, we say that \\
(1) $\alpha$ is {\it enabled by} $T$ if $R_\alpha \subseteq T$ and $I_\alpha \cap T = \emptyset$, \\
(2)  $\alpha$ is {\it enabled by $T$ in maximally parallel manner}   
if there is no  $\beta \in A^\#$ such that $\alpha \subset \beta$,  and  $\alpha$ and $\beta$  are enabled by  $T$.   \\
(3)  By $En^p_A(T)$ we denote the set of all multisets of reactions $\alpha \in A^\#$ which are enabled by $T$ in maximally parallel manner.\\
(4) The {\it results of $A$ on $T$}, denoted by $Res_{A}(T)$, is defined as follows: 
\[ 
Res_{A}(T) = \{ T - R_{\alpha} + P_{\alpha} \, | \, \alpha \in En^p_A(T)  \}. 
\] 
Note that we have $Res_{A}(T)= \{ T \}$ if $En^p_A(T)=\emptyset$. Thus, if no multiset of reactions $\alpha \in A^{\#}$ is enabled by $T$ in maximally parallel manner, then $T$ remains unchanged.  
}
\end{de}

\begin{rem}
{\rm 
 ($i$)\ It should be also noted that the definition of the results of $A$ on $T$ (given in (4) above) is in contrast to the original one in \cite{ER:07a},  because 
we  adopt the assumption of {\it permanency of elements}: 
any element that is not a reactant for any active reaction {\it does} remain in the result after the reaction.\\
($ii$)\ In general, $En^p_A(T)$ may contain more than one element, and 
therefore, so may $Res_A(T)$.\\
($iii$)\ For simplicity, $I_a$ is often represented as a string rather than a set.
}
\end{rem}

\begin{exam}{\rm 
Let $S=\{a, b, c, d, e\}$ and consider the following set $A=\{{\bf a}, {\bf b}, {\bf c}\}$ of reactions in $S$:    
\[
 {\bf a} = ( b^2, a, c ),\, {\bf b}= ( c^2, \emptyset, b ), \, {\bf c} = ( bc, d, e).  
\]
$(i)$\ Consider a finite multiset $T=b^4cd$. Then,  $\alpha_1={\bf a}$ is enabled by $T$, while neither {\bf b} nor {\bf c} is enabled by $T$, because $R_{{\bf b}} \not\subseteq T$ and $I_{{\bf c}} \cap T\not= \emptyset$.  Further, 
 $\alpha_2={\bf a}^2$ is not only enabled by $T$ but also 
enabled by $T$ in maximally parallel manner, because no $\beta$ with $\alpha_2\subset \beta$ is enabled by $T$.   Since $R_{{\bf a}^2}=b^4$, $P_{{\bf a}^2}=c^2$, and 
 $En^p_A(T)=\{{\bf a}^2\}$,  we have 
\[
Res_A(T)=\{ T-R_{{\bf a}^2}+P_{{\bf a}^2} \} =\{c^3d\}. 
\]
$(ii)$\ Consider $T'=b^3c^2e$. Then,  $\beta_1={\bf ab}$ and $\beta_2={\bf ac}$ 
are enabled by $T'$, while   {\bf bc} is not.  Further, 
 it is seen that both $\beta_1$ and $\beta_2$ are enabled by $T'$ in 
maximally parallel manner, and  $En^p_A(T')=\{{\bf ab}, {\bf ac}\}$. 
Thus, we have 
\[
Res_A(T')= \{b^2ce, c^2e^2\}. 
\]
If we take $T''=bcd$, then none of the reactions from $A$ is enabled by $T''$. Therefore, we have $Res_A(T'')=T''$.
}
\end{exam}

We are now in a position to introduce the notion of reaction automata.

\begin{de}{\rm 
{(Reaction Automata)}\ A {\it reaction automaton} (RA) $\mathcal{A}$ is a 5-tuple $\mathcal{A} = (S, \Sigma, A, D_0, f)$, where
\begin{itemize}
\item $S$ is a finite set,  called the {\it background set of}  $\mathcal{A}$,
\item $\Sigma (\subseteq S)$ is called the {\it input alphabet of}  $\mathcal{A}$, 
\item $A$ is a finite set of reactions in $S$,
\item $D_0 \in S^\#$ is an {\it initial multiset},
\item $f \in S$ is a special symbol which indicates the final state.
\end{itemize}
}
\end{de}

\begin{de}{\rm 
Let $\mathcal{A} = (S, \Sigma, A, D_0, f)$ be an RA and $w = a_1 \cdots a_n \in \Sigma^*$.  An {\it interactive process in $\mathcal{A}$ with input $w$} is an infinite sequence  
$\pi = D_0, \dots, D_i, \dots$, where 
\begin{eqnarray*}
\left\{ \begin{array}{ll}
 D_{i+1} \in Res_A(a_{i+1}+D_i)   & \mbox{(for $0\leq i \leq n-1$), and} \\
 D_{i+1} \in Res_A(D_i) & \mbox{(for all $i\geq n$)}.
\end{array} \right.
\end{eqnarray*}
By $IP(\mathcal{A}, w)$ we denote the set of all interactive processes in $\mathcal{A}$ with input $w$.}
\end{de}
In order to represent an interactive process $\pi$, we also use 
the ``arrow notation'' for $\pi : (a_1, D_0) \rightarrow \cdots \rightarrow (a_{n}, D_{n-1}) \rightarrow (D_n) \rightarrow (D_{n+1}) \rightarrow \cdots$, 
or alternatively, 
 $D_0 \rightarrow^{a_1}  D_1 \rightarrow^{a_2} D_2 \rightarrow^{a_3} 
\cdots \rightarrow^{a_{n-1}}  D_{n-1} \rightarrow^{a_{n}}  D_{n}  \rightarrow D_{n+1} \rightarrow \cdots$.

For an interactive process $\pi$ in $\mathcal{A}$ with input $w$, if $En^p_A(D_m) = \emptyset$ for some $m \ge |w|$, then we have that  $Res_A(D_m)= \{ D_m \}$ and $D_m =D_{m+1}=
\cdots$. In this case, considering the smallest $m$, we say that $\pi$ {\it converges on} $D_m$ 
(at the $m$-th step). When an interactive process $\pi$ converges 
on $D_m$, each $D_i$ of $\pi$ is omitted for $i \ge m+1$.

\begin{de}{\rm 
Let $\mathcal{A} = (S, \Sigma, A, D_0, f)$ be an RA.  The {\it language accepted by} $\mathcal{A}$, denoted by $L(\mathcal{A})$, is defined as follows:
\begin{align*}
L(\mathcal{A}) = \{ w \in \Sigma^* \, | \, & \,  \text{there exists }\pi \in IP(\mathcal{A}, w) \text{ that converges on} \\
& \text{$D_m$ at the $m$-th step for some $m\geq |w|$, }
\text{and } f \subseteq D_m \}. 
\end{align*}}
\end{de}

\begin{figure}[t]
\centerline{
\includegraphics[scale=0.55]{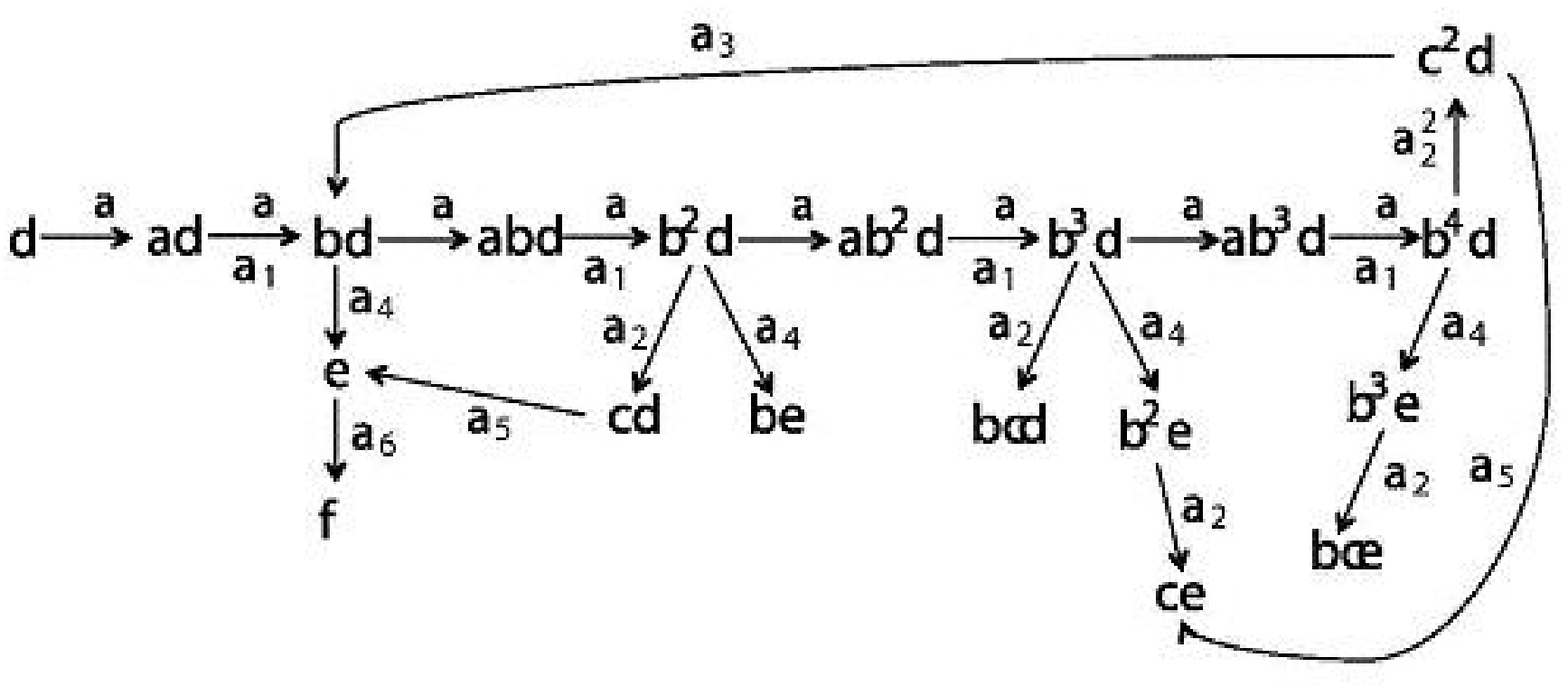}}
\caption{(a) Reaction diagram: Interactive processes for accepting $a^2$, $a^4$ and $a^8$ in $\mathcal{A}$. Some arrows are associated with a multiset of reactions applied at the step.}
\label{diag}
\end{figure}

\begin{exam}{\rm 
Let us consider a reaction automaton $\mathcal{A} = (S, \Sigma, A, D_0, f)$ defined as follows:
\begin{align*}
&S = \{ a, b, c, d, e, f \} \mbox{ with }  \Sigma=\{a\},  \\
&A = \{ {\bf a}_1, {\bf a}_2, {\bf a}_3, {\bf a}_4, {\bf a}_5, {\bf a}_6 \}, \ \mbox{where} \\
&\text{\ \ \ \ \ \ } {\bf a}_1 = ( a^2, \emptyset, b ),\ \,{\bf a}_2 = ( b^2, ac, c ),\, {\bf a}_3 = ( c^2, b, b ), \\
&\qquad  {\bf a}_4 = ( bd, ac, e ), \, {\bf a}_5 = ( cd, b, e ),\ \ {\bf a}_6 = ( e, abc, f ), \\
&D_0 = d.
\end{align*}
Let $w = aaaaaaaa \in S^*$ be the input string and consider an interactive process $\pi$ such that
\begin{align*}
\pi : d \rightarrow^a  ad \rightarrow^a  bd  \rightarrow^a 
 abd \rightarrow^a  b^2d  \rightarrow^a ab^2d \rightarrow^a b^3d  \rightarrow ^a ab^3d  \rightarrow^a b^4d \rightarrow c^2d \rightarrow bd \rightarrow e \rightarrow f. 
\end{align*}
It can be easily seen that $\pi \in IP(\mathcal{A}, w)$ and $w \in L(\mathcal{A})$. Figure \ref{diag} illustrates the whole view of possible interactive processes  in $\mathcal{A}$ with inputs 
$a^2, a^4$ and $a^8$.
For instance, since ${\bf a}^2_2  \in En^p_A(b^4 d)$, it holds that $c^2d \in Res_{A}(b^4 d)$. Hence, the step $b^4 d \rightarrow c^2d$ is valid. 
We can also see that $L(\mathcal{A}) = \{ a^{2^n} \, | \, n \ge 1 \}$ which  is context-sensitive.
}
\end{exam}

\section{Multistack Machines}
A multistack machine is a deterministic pushdown automaton with several stacks (\cite{HMU:03}).  It is known that a two-stack machine is equivalent to a Turing machine as a language accepting device.

A $k$-stack machine $M = (Q, \Sigma, \Gamma, \delta, p_0, Z_0, F)$ is defined as follows: $Q$ is a set of states, $\Sigma$ is an input alphabet, $\Gamma$ is a stack alphabet, $Z_0=(Z_{01},Z_{02},\ldots,Z_{0k})$ is the $k$-tuple of the initial stack symbols, $p_0 \in Q$ is the initial state, $F$ is a set of final states, $\delta$ is a transition function defined in the form: $\delta(p, a, X_1, X_2, \dots, X_k) = (q, \gamma_1, \gamma_2, \ldots, \gamma_k)$,
where $p,q \in Q$, $a \in \Sigma \cup \{ \lambda \}$, $X_i \in \Gamma$, $\gamma_i \in \Gamma^*$ for each $1 \le i \le k$. This rule means that in state $p$, with $X_i$ on the top of $i$-th stack, if the machine reads $a$ from the input, then go to state $q$, and replace the the top of each $i$-th stack with $\gamma_i$ for $1 \le i \le k$. We assume that each rule has a unique label and all labels of rules in $\delta$ is denoted by $Lab(\delta)$.
Note that the $k$-stack machine can make a $\lambda$-move, but there cannot be a choice of a $\lambda$-move or a non-$\lambda$-move due to the deterministic property of the machine. The $k$-stack machine accepts a string by entering a final state.

In this paper, we consider a modification on a multistack (in fact, two-stack) machine. Recall that in the simulation of a given Turing machine $TM$ with an input $w=a_1a_2\cdots a_{\ell}$ in terms of a multistack 
machine $M$, one can assume the following (see \cite{HMU:03}):
\begin{itemize}
\item[($i$)] At first, two-stack machine $M$ is devoted to making the copy of $w$ on stack-2. This is illustrated in (a) and (b)-1 of Figure \ref{sim}, for the case of $k=2$. {\it $M$ requires only non-$\lambda$-moves}.  
\item[($ii$)] Once the whole input $w$ is read-in by $M$, {\it no more access to the input tape of $M$ is necessary}. After having $w^R$ on stack-2, $M$ moves over $w^R$ (from stack-2) to produce $w$ on stack-1, 
as shown in (b)-2. These moves only require $\lambda$-moves and after this, each computation step of $M$ with respect to $w$ is performed by a $\lambda$-move, without any access to $w$ on the input tape.   
\item[($iii$)] Each stack has its own stack alphabet, each one being different from the others,  and a set of final states is a singleton. Once $M$ enters the final state, it immediately halts. Further, during a computation, each stack is not emptied. 
\end{itemize}
Hence, without changing the computation power, we may restrict all computations of a multistack machine that satisfies the conditions $(i), (ii), (iii)$.  We call this modified multistack machine a {\it restricted multistack machine}.  

In summary, a restricted $k$-stack machine $M_r$ is composed by $2k+5$ elements as follows:
\[ M_r = (Q, \Sigma, \Gamma_1, \Gamma_2, \dots, \Gamma_k, \delta, p_0, Z_{01}, Z_{02}, \dots, Z_{0k}, f), \]
where for each $1\leq i \leq k$, $Z_{0i} \in \Gamma_i$ is the initial  symbol for the $i$-th stack used only for the bottom, $f \in Q$ is a final state, and its computation proceeds only in the above mentioned way ($i$), ($ii$), ($iii$). Especially, $\lambda$-moves are used  after all non-$\lambda$-moves in a computation of $M_r$.

\begin{prop}
{\rm (Theorem 8.13 in \cite{HMU:03})} 
Every recursively enumerable language is accepted by a restricted two-stack machine.
\end{prop}

\begin{figure}[t]
\centerline{
\includegraphics[scale=0.26]{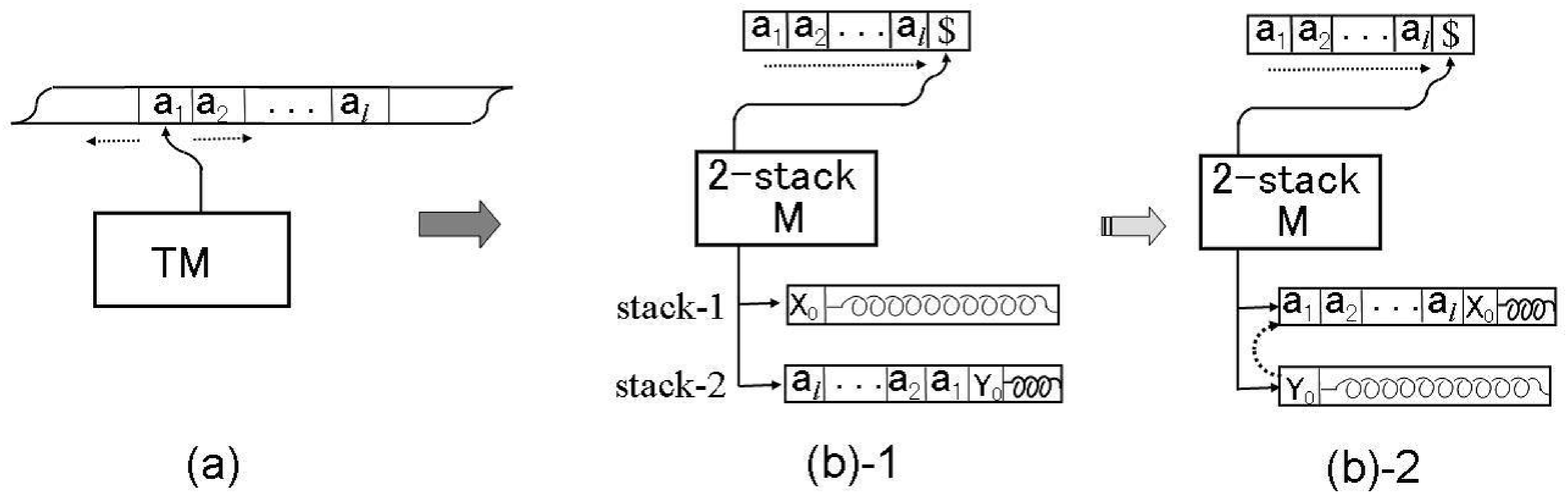}}
\caption{(a) Turing machine (TM);  (b)Two-stack machine $M$ simulating TM, where  $\$$ is the end marker for the input.}
\label{sim}
\end{figure}

\section{Main Results}

In this section we shall show the equivalence of the accepting powers between reaction machines and Turing machines. Taking Proposition 1 into 
consideration, it should be enough for the purpose of this paper to prove the following theorem.

\begin{thm}
If a language $L$ is accepted by a restricted two-stack machine, then $L$ is accepted by a reaction automaton. 
\label{teiri1}
\end{thm}
\noindent
\textbf{[Construction of an RA]}\\[2mm]
Let $M = (Q, \Sigma, \Gamma_1, \Gamma_2, \delta, p_0, {X}_0, {Y}_0, f)$  be a restricted two-stack machine with  
 $\Gamma_1 = \{ X_0, X_1, \dots, X_n \}$, $\Gamma_2 = \{ Y_0, Y_1, \dots, Y_m \}$, $n, m \ge 1$, where $\Gamma = \Gamma_1 \cup \Gamma_2$, $X_0$ and  $Y_0$ are the initial stack symbols for stack-1 and stack-2, repsectively,  
and we may assume that $\Gamma_1 \cap \Gamma_2 =\emptyset$. 

We construct an RA $\mathcal{A}_M = (S, \Sigma, A, D_0, f')$ as follows:
\begin{align*}
&S = Q \cup \hat{Q} \cup \Sigma \cup \Gamma \cup \hat{\Gamma} \cup Lab(\delta) \cup \{ f' \}, \\
&A = A_0 \cup A_a \cup \hat{A_a}  \cup A_{\lambda} \cup \hat{A}_{\lambda} \cup A_{X} \cup \hat{A}_{X} \cup A_{Y} \cup \hat{A}_{Y} \cup  A_f \cup \hat{A}_f, \\
&D_0 = p_0 X_0 Y_0, 
\end{align*}
where the set of reactions $A$ consists of the following 5 categories : 
\begin{align*}
(1) &A_0 =\{ ( p_0 a X_0 Y_0, Lab(\delta),  \hat{q} \cdot stm(\hat{x}) \cdot stm(\hat{y}) \cdot r' ) \, | \, r: \delta(p_0, a, X_0, Y_0) = (q, x, y), \, r' \in Lab(\delta) \},  \\
(2) &A_a = \{ ( p a X_i Y_j r, \hat{\Gamma}, \hat{q} \cdot stm(\hat{x}) \cdot stm(\hat{y}) \cdot r' ) \, | \, a \in \Sigma, \, r: \delta(p, a, X_i, Y_j) = (q, x, y), \, r' \in Lab(\delta) \},  \\
&\hat{A}_a = \{ ( \hat{p} a \hat{X_i} \hat{Y_j} r, \Gamma, q \cdot stm(x) \cdot stm(y) \cdot r' ) \, | \, a \in \Sigma, \, r: \delta(p, a, X_i, Y_j) = (q, x, y), \, r' \in Lab(\delta) \},  \\
(3) &A_{\lambda} = \{ ( p X_i Y_j r, \Sigma \cup \hat{\Gamma}, \hat{q} \cdot stm(\hat{x}) \cdot stm(\hat{y}) \cdot r' ) \, | \, r: \delta(p, \lambda, X_i, Y_j) = (q, x, y), \, r' \in Lab(\delta) \},  \\
&\hat{A}_{\lambda} = \{ ( \hat{p} \hat{X_i} \hat{Y_j} r, \Sigma \cup \Gamma, q \cdot stm(x) \cdot stm(y) \cdot r' ) \, | \, r: \delta(p, \lambda, X_i, Y_j) = (q, x, y), \, r' \in Lab(\delta) \},  \\
(4) &A_{X} = \{ ( X^2_k, \hat{Q} \cup \hat{\Gamma} \cup (Lab(\delta) - \{ r \}) \cup \{ f' \},  \hat{X}^{2^{|x|}}_k ) \, | \, 0 \le k \le n, \, r: \delta(p, a, X_i, Y_j) = (q, x, y) \},  \\
&\hat{A}_{X} = \{ ( \hat{X}^2_k, Q \cup \Gamma \cup (Lab(\delta) - \{ r \}) \cup \{ f' \},  X^{2^{|x|}}_k ) \, | \, 0 \le k \le n, \, r: \delta(p, a, X_i, Y_j) = (q, x, y) \},  \\
&A_{Y} = \{ ( Y^2_k, \hat{Q} \cup \hat{\Gamma} \cup (Lab(\delta) - \{ r \}) \cup \{ f' \},  \hat{Y}^{2^{|y|}}_k ) \, | \, 0 \le k \le m, \, r: \delta(p, a, X_i, Y_j) = (q, x, y) \},  \\
&\hat{A}_{Y} = \{ ( \hat{Y}^2_k, Q \cup \Gamma \cup (Lab(\delta) - \{ r \}) \cup \{ f' \}, Y^{2^{|y|}}_k ) \, | \, 0 \le k \le m, \, r: \delta(p, a, X_i, Y_j) = (q, x, y) \},  \\
(5) &A_f = \{ ( f, \hat{\Gamma}, f' ) \},  \\
&\hat{A}_f = \{ ( \hat{f}, \Gamma, f' ) \}.  
\end{align*}

\begin{proof}
We shall give an informal description on how to simulate $M$ with 
an input $w=a_1a_2\cdots a_{\ell}$ in terms of $\mathcal{A}_M$ constructed above. 

$M$ starts its computation from the state $p_0$ with $X_0$ and $Y_0$ on the top of stack-1 and stack-2, respectively. This initial step 
is performed in $\mathcal{A}_M$ by applying a reaction in $A_0$ to $D_0=p_0X_0Y_0$ together with $a_1$.  In order to read the whole input $w$ into $\mathcal{A}_M$, applying reactions in (2) and (4) leads 
to an interactive process in $\mathcal{A}_M$ :  $D_0 \rightarrow^{a_1} D_1 \rightarrow^{a_2} D_2 
\rightarrow^{a_3} \cdots \rightarrow^{a_\ell} D_{\ell}$, where $D_{\ell}$  just corresponds to  
the configuration of $M$ depicted in (b)-1 of Figure \ref{sim}.  After this point, only reactions from (3), (4) and (5) are  available in $\mathcal{A}_M$, because  $M$ makes only $\lambda$-moves.

Suppose that for $k\geq 1$,   after making $k$-steps $M$ is in the state 
$p$ and has  $\alpha_k \in \Gamma_1^*$ and $\beta_k \in \Gamma_2^*$ 
on the stack-1 and the stack-2, respectively. Then,   from the manner of constructing $A$, it is seen  
that in the corresponding interactive process in $\mathcal{A}_M$,  we have :
\begin{eqnarray*}
\left\{ \begin{array}{ll}
 D_k = p \cdot stm(\alpha_k) \cdot stm(\beta_k) \cdot r & \mbox{(if $k$ is even)} \\
 D_k = \hat{p} \cdot stm(\hat{\alpha}_k) \cdot stm(\hat{\beta}_k) \cdot r & \mbox{(if $k$ is odd)} \\
\end{array} \right.
\end{eqnarray*}
for some $r \in Lab(\delta)$, where the rule labeled by $r$ may be used at the $(k+1)$-th step. 
(Recall that $stm(x)$ is a multiset,  in a special 2-power form, representing a string $x$.)  
Thus,  the multisubset  ``$stm(\alpha_k)stm(\beta_k)$'' in $D_k$ is  denoted by the strings in either 
$\Gamma^*$ or $\hat{\Gamma}^*$ in an alternate fashion, depending upon the value $k$.  
 Since there is no essential difference between strings denoted by $\Gamma^*$ and its hat version, 
we only argue about the case when $k$ is even.

Suppose that $M$ is in the state $p$ and has $\alpha=X_{i1}\cdots X_{it}X_0$  on the  stack-1 and 
$\beta=Y_{j1}\cdots Y_{js}Y_0$  on the  stack-2, where the leftmost element is the top symbol of the stack. Further,  
let $r$ be the label of a transition $\delta(p, a_{k+1}, X_{i1},Y_{j1})=(q,x,y)$ (if $1 \le k \le l-1$) or $\delta(p,\lambda, X_{i1},Y_{j1})=(q,x,y)$ (if $l \le k$) in $M$ to be applied.  Then, the two stacks are 
 updated as $\alpha'=x X_{i2}\cdots X_{it}X_0$  and  $\beta'=y Y_{j2}\cdots Y_{js}Y_0$. 
 In order to simulate this move  of $M$, we need to prove that  
it is possible in $\mathcal{A}_M$, $D_k \rightarrow^{a_{k+1}} D_{k+1}$ (if $1 \le k \le l-1$) or $D_k \rightarrow D_{k+1}$ (if $l \le k$), where
\begin{align*}
&D_k = p\cdot stm(X_{i1} {X}_{i2}\cdots {X}_{it}{X_0}) \cdot stm({Y_{j1}} {Y}_{j2}\cdots {Y}_{js}{Y_0}) r  \\
&D_{k+1} = \hat{q}\cdot stm(\hat{x} \hat{X}_{i2}\cdots \hat{X}_{it}\hat{X_0}) \cdot stm(\hat{y} \hat{Y}_{j2}\cdots \hat{Y}_{js}\hat{Y_0}) r' 
\end{align*}
for some $r'\in Lab(\delta)$.  Taking a close look at $D_k$, we have that
\[  
D_k=p X_{i1}Y_{j1} r \cdot X^2_{i2} {X}^{2^2}_{i3} \cdots {X}^{2^{t-1}}_{it}{X}^{2^t}_0 \cdot  Y^2_{j2} {Y}^{2^2}_{j3} \cdots {Y}^{2^{s-1}}_{js} {Y}^{2^s}_0,
\]
from which it is easily seen that a multiset of reactions ${\bf z}=
{\bf r x_{i2}\cdots x}_{\bf it}^{2^{t-2}} {\bf x}_{\bf 0}^{2^{t-1}} {\bf y_{j2} \cdots y}_{\bf js}^{2^{s-2}} {\bf y}_{\bf 0}^{2^{s-1}}$ is in $En^p_{\mathcal{A}_M}(a_{k+1} + D_k)$ (if $1 \le k \le l-1$) or in $En^p_{\mathcal{A}_M}(D_k)$ (if $l \le k$), 
i.e., it is enabled by $a_{k+1} + D_k$ (if $1 \le k \le l-1$) or $D_k$ (if $l \le k$) in maximally parallel manner,  where
\begin{eqnarray*}
\left\{ \begin{array}{ll}
 {\bf r}&=(p a_{k+1} X_{i1}Y_{j1} r, \hat{\Gamma}, \hat{q} \cdot stm(\hat{x})stm(\hat{y}) r') \in A_{a} \mbox{ (if $1 \le k \le l-1$)} \\
 {\bf r}&=(p X_{i1}Y_{j1} r, \Sigma\cup \hat{\Gamma}, \hat{q} \cdot stm(\hat{x})stm(\hat{y}) r') \in A_{\lambda} \mbox{ (if $l \le k$)}, \\
\end{array} \right.
\end{eqnarray*} 
for some $r' \in Lab(\delta)$, \\[-6mm]
\begin{align*}
{\bf x}_i&=(X^2_{i},\hat{Q}\cup \hat{\Gamma} \cup Lab(\delta)-\{r\} \cup \{ f' \}, \hat{X}^{2^{|x|}}_i) \in A_X \mbox{ (for $i=0, i2,\ldots, it$)}, \\
{\bf y}_j&=(Y^2_{j},\hat{Q} \cup \hat{\Gamma} \cup Lab(\delta)-\{r\} \cup \{ f' \}, \hat{Y}^{2^{|y|}}_j) \in A_Y \mbox{ (for $j=0, j2,\ldots, js$)}.
\end{align*}
The result of the multiset of the reactions ${\bf z}$ is 
\begin{align*}  
&\hat{q} \cdot stm(\hat{x})stm(\hat{y})r'\cdot \hat{X}^{2^{|x|}}_{i2}  \cdots \hat{X}^{2^{t-2+|x|}}_{it}\hat{X}^{2^{t-1+|x|}}_0 \cdot  \hat{Y}^{2^{|x|}}_{j2}  \cdots \hat{Y}^{2^{s-2+|x|}}_{js} \hat{Y}^{2^{s-1+|x|}}_0 \\
= \, &\hat{q} \cdot stm(\hat{x} \hat{X}_{i2}\cdots \hat{X}_{it}\hat{X_0}) \cdot stm(\hat{y} \hat{Y}_{j2}\cdots \hat{Y}_{js}\hat{Y_0}) r' \\
= \, &D_{k+1}
\end{align*}
Thus, in fact it holds that $D_k \rightarrow^{a_{k+1}} D_{k+1}$ (if $1 \le k \le l-1$) or $D_k \rightarrow D_{k+1}$ (if $l \le k$) in $\mathcal{A}_{M}$. 

We note that there is a possibility that undesired reaction  ${\bf r'}$ can be enabled at the $(k+1)$th step, where ${\bf r'}$ is of the form
\begin{eqnarray*}
\left\{ \begin{array}{ll}
 {\bf r'}&=(p a_{k+1} X_{iu}Y_{jv} r, \hat{\Gamma}, \hat{q'} \cdot stm(\hat{x'})stm(\hat{y'}) r') \in A_{a} \mbox{ (if $1 \le k \le l-1$)}\\
 {\bf r'}&=(p X_{iu}Y_{jv} r, \Sigma\cup \hat{\Gamma}, \hat{q'} \cdot stm(\hat{x'})stm(\hat{y'}) r') \in A_{\lambda} \mbox{ (if $l \le k$)}, \\
\end{array} \right.
\end{eqnarray*}
with $u \ne 1$ or $v \ne 1$, that is, the reactant of ${\bf r'}$ contains a stack symbol which is not the top of stack. If a multiset of reactions ${\bf z'}={\bf r' x'_1\cdots x'_{t'} y'_1 \cdots y'_{s'}}$ with ${\bf x'_1, \ldots, x'_{t'}} \in A_X$, ${\bf y'_1, \ldots, y'_{s'}} \in A_Y$ is used at the $(k+1)$th step, then $D_{k+1}$ contains {\it both} the symbols without hat (in $\Gamma$) and the symbols with hat (in $\hat{Q}$ and $\hat{\Gamma}$). This is because in this case, $X_{i1}$ or $Y_{j1}$ in $D_k$ which is not consumed at the $(k+1)$-th step remains in $D_{k+1}$ (since the total numbers of $X_{i1}$ and $Y_{j1}$ are {\it odd}, these objects cannot be consumed out by the reactions from (4)). Hence, no reaction is enabled at the $(k+2)$-th step and $f'$ is never derived after this wrong step.

From the arguments above, it holds that for  an input $w \in \Sigma^*$,
 $M$ enters the final state $f$  (and halts) 
 if and only if there exists $\pi : D_0, \ldots ,D_i, \ldots \in IP(\mathcal{A}_M,w)$ such that $D_{k-1}$ contains $f$ or $\hat{f}$, $D_k$ contains $f'$, and $\pi$ converges on $D_k$, for some $k \ge 1$. 
Therefore,  we have that $L(M) = L(\mathcal{A}_M)$ holds.  

\end{proof}

\begin{cor}
Every recursively enumerable language is accepted by a reaction automaton.
\end{cor}

Recall the way of constructing reactions $A$ of $\mathcal{A}_M$ in the proof of Theorem 1. The reactions in categories (1), (2), (3) would not satisfy the condition of determinacy which is given immediately below. However, we can easily modify $\mathcal{A}_M$ to meet the condition.

\begin{de}{\rm 
Let $\mathcal{A}_M = (S, \Sigma, A, D_0, f')$ be an RA. Then, $\mathcal{A}_M$ is 
{\it deterministic} if for $a=(R, I, P), a' =(R', I', P') \in A$,  $(R = R') \wedge (I = I')$ implies that $a = a'$.
}
\end{de}

\begin{thm}
If a language $L$ is accepted by a restricted two-stack machine, then $L$ is accepted by a deterministic reaction automaton.
\end{thm}

\begin{proof}
Let $M = (Q, \Sigma, \Gamma_1, \Gamma_2, \delta, p_0, {X}_0, {Y}_0, f)$  be a restricted two-stack machine. For the RA $\mathcal{A}_M = (S, \Sigma, A, D_0, f')$ constructed for the proof of Theorem 1, we consider $\mathcal{A}'_M = (S \cup \hat{Lab(\delta)}, \Sigma, A', D_0, f')$, where $A'$ consists of the following 5 categories : 
\begin{align*}
(1) &A_0 =\{ ( p_0 a X_0 Y_0, Lab(\delta) \cup \{ \hat{r'} \},  \hat{q} \cdot stm(\hat{x}) \cdot stm(\hat{y}) \cdot \hat{r'} ) \, | \, r: \delta(p_0, a, X_0, Y_0) = (q, x, y), \, r' \in Lab(\delta) \},  \\
(2) &A_a = \{ ( p a X_i Y_j r, \hat{\Gamma} \cup \{ \hat{r'} \}, \hat{q} \cdot stm(\hat{x}) \cdot stm(\hat{y}) \cdot \hat{r'} ) \, | \, a \in \Sigma, \, r: \delta(p, a, X_i, Y_j) = (q, x, y), \, r' \in Lab(\delta) \},  \\
&\hat{A}_a = \{ ( \hat{p} a \hat{X_i} \hat{Y_j} r, \Gamma \cup \{ r' \}, q \cdot stm(x) \cdot stm(y) \cdot r' ) \, | \, a \in \Sigma, \, r: \delta(p, a, X_i, Y_j) = (q, x, y), \, r' \in Lab(\delta) \},  \\
(3) &A_{\lambda} = \{ ( p X_i Y_j r, \Sigma \cup \hat{\Gamma} \cup \{ \hat{r'} \}, \hat{q} \cdot stm(\hat{x}) \cdot stm(\hat{y}) \cdot \hat{r'} ) \, | \, r: \delta(p, \lambda, X_i, Y_j) = (q, x, y), \, r' \in Lab(\delta) \},  \\
&\hat{A}_{\lambda} = \{ ( \hat{p} \hat{X_i} \hat{Y_j} r, \Sigma \cup \Gamma \cup \{ r' \}, q \cdot stm(x) \cdot stm(y) \cdot r' ) \, | \, r: \delta(p, \lambda, X_i, Y_j) = (q, x, y), \, r' \in Lab(\delta) \},  \\
(4) &A_{X} = \{ ( X^2_k, \hat{Q} \cup \hat{\Gamma} \cup (\hat{Lab(\delta)} - \{ \hat{r} \}) \cup \{ f' \},  \hat{X}^{2^{|x|}}_k ) \, | \, 0 \le k \le n, \, r: \delta(p, a, X_i, Y_j) = (q, x, y) \},  \\
&\hat{A}_{X} = \{ ( \hat{X}^2_k, Q \cup \Gamma \cup (Lab(\delta) - \{ r \}) \cup \{ f' \},  X^{2^{|x|}}_k ) \, | \, 0 \le k \le n, \, r: \delta(p, a, X_i, Y_j) = (q, x, y) \},  \\
&A_{Y} = \{ ( Y^2_k, \hat{Q} \cup \hat{\Gamma} \cup (\hat{Lab(\delta)} - \{ \hat{r} \}) \cup \{ f' \},  \hat{Y}^{2^{|y|}}_k ) \, | \, 0 \le k \le m, \, r: \delta(p, a, X_i, Y_j) = (q, x, y) \},  \\
&\hat{A}_{Y} = \{ ( \hat{Y}^2_k, Q \cup \Gamma \cup (Lab(\delta) - \{ r \}) \cup \{ f' \},  Y^{2^{|y|}}_k ) \, | \, 0 \le k \le m, \, r: \delta(p, a, X_i, Y_j) = (q, x, y) \},  \\
(5) &A_f = \{ ( f, \hat{\Gamma}, f' ) \},  \\
&\hat{A}_f = \{ ( \hat{f}, \Gamma, f' ) \}.  
\end{align*}

The reactions in categories (1), (2), (3) in $A'$ meet the condition where $\mathcal{A'}_M$ is deterministic, since the inhibitor of each reaction includes $r'$ or $\hat{r'}$. We can easily observe that the equation $L(M) = L(\mathcal{A'}_M)$ is proved in the manner similar to the proof of Theorem 1.

\end{proof}  

\begin{cor}
Every recursively enumerable language is accepted by a deterministic reaction automaton.
\end{cor}

\section{Space Complexity Classes of RAs}

We now consider space complexity issues of reaction automata. That is, we introduce some subclasses of reaction automata and investigate the relationships between classes of languages accepted by those subclasses of automata and language classes in the Chomsky hierarchy.

Let $\mathcal{A}$ be an RA and $f$ be a function defined on $\mathbf{N}$.  
Motivated by the notion of a workspace for a phrase-structure grammar (\cite{AS:73}), we define: for $w\in L(\mathcal{A})$ with $n=|w|$, and for $\pi$ in $IP(\mathcal{A},w)$  that  converges on  $D_m$ for some $m \geq n$ and $D_m$ includes the final state, 
\[
WS(w,\pi) = \underset{i}{{\rm max}} \{|D_i| \mid  D_i \mbox{ appears in } \pi \ \}. 
\]
Further, the {\it workspace of} $\mathcal{A}$ {\it for} $w$ is defined as:
\[
WS(w,\mathcal{A}) =\underset{\pi}{{\rm min}} \{WS(w,\pi)  \mid \pi \in IP(\mathcal{A},w) \}.  
\]

\begin{de}{\rm 
($i$). An RA  $\mathcal{A}$ is {\it $f(n)$-bounded} if for any $w\in L(\mathcal{A})$ with $n=|w|$,   $WS(w,\mathcal{A})$ is bounded by $f(n)$. \\
($ii$). 
If a function $f(n)$ is a constant $k$ (resp. linear, polynomial,  exponential), 
then  $\mathcal{A}$ is termed $k$-bounded (resp. linearly-bounded,  polynomially-bounded,  exponentially-bounded), and denoted by 
$k$-RA (resp. $lin$-RA, $poly$-RA, $exp$-RA). Further, 
the class of languages accepted by $k$-RA (resp. $lin$-RA,  $poly$-RA, $exp$-RA, arbitrary RA) is denoted by $k$-$\mathcal{RA}$ 
(resp. $\mathcal{LRA, PRA, ERA, RA}$). 
}
\end{de}

Let us denote by $\mathcal{REG}$ (resp. $\mathcal{LIN, CF, CS, RE}$) 
 the class of regular (resp. linear context-free, context-free, context-sensitive, recursively enumerable) languages.

\begin{exam}{\rm 
Let $L_1 = \{ a^n b^n c^n \, | \, n \ge 0 \}$ and consider an RA $\mathcal{A}_1 = (S, \Sigma, A, D_0, f)$ defined as follows:
\begin{align*}
&S = \{ a, b, c, d, a', b', c', f \} \mbox{ with }  \Sigma=\{a, b, c\},  \\
&A = \{ {\bf a}_1, {\bf a}_2, {\bf a}_3, {\bf a}_4 \}, \ \mbox{where} \\
&\text{\ \ \ \ \ \ } {\bf a}_1 = ( a, bb', a' ), \,{\bf a}_2 = ( a'b, cc', b' ),\, {\bf a}_3 = ( b'c, \emptyset, c' ), \, {\bf a}_4 = ( d, abca'b', f ), \\
&D_0 = d.
\end{align*}
Then, it holds that $L_1 = L(\mathcal{A}_1)$ (see Figure~\ref{ra-anbncn}).
}
\end{exam}

\begin{exam}{\rm 
Let $L_2 = \{ a^m b^m c^n d^n \, | \, m, n \ge 0 \}$ and consider an RA $\mathcal{A}_2 = (S, \Sigma, A, D_0, f)$ defined as follows:
\begin{align*}
&S = \{ a, b, c, d, a', c', p_0, p_1, p_2, p_3, f \} \mbox{ with }  \Sigma=\{a, b, c, d \},  \\
&A = \{ {\bf a}_1, {\bf a}_2, {\bf a}_3, {\bf a}_4, {\bf a}_5, {\bf a}_6, {\bf a}_7, {\bf a}_8, {\bf a}_9, {\bf a}_{10}, {\bf a}_{11} \}, \ \mbox{where} \\
&\text{\ \ \ \ \ \ } {\bf a}_1 = ( ap_0, bc, a'p_0 ), \,{\bf a}_2 = ( a'bp_0, c, p_1 ),\, {\bf a}_3 = ( a'bp_1, c, p_1 ), \, {\bf a}_4 = ( cp_0, d, c'p_2 ), \\
&\text{\ \ \ \ \ \ } {\bf a}_5 = ( cp_1, d, c'p_2 ), \, {\bf a}_6 = ( cp_2, d, c'p_2 ), {\bf a}_7 = ( c'dp_2, \emptyset, p_3 ), \,{\bf a}_8 = ( c'dp_3, \emptyset, p_3 ), \\
&\text{\ \ \ \ \ \ } {\bf a}_9 = ( p_0, abcd, f ), \, {\bf a}_{10} = ( p_1, abcda', f ), \, {\bf a}_{11} = ( p_3, abcda'c', f ), \\
&D_0 = p_0.
\end{align*}
Then, it holds that $L_2 = L(\mathcal{A}_2)$ (see Figure~\ref{ra-anbncmdm}).
}
\end{exam}
It should be noted that $\mathcal{A}_1$ and $\mathcal{A}_2$ are both 
$lin$-RAs, therefore, the class of languages $\mathcal{LRA}$ includes 
a context-sensitive language $L_1$  and a non-linear context-free language $L_2$.    

\begin{figure}[t]
\centerline{
\includegraphics[scale=0.63]{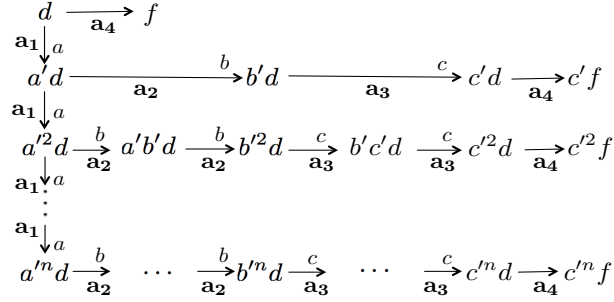}}
\caption{Reaction diagram of $\mathcal{A}_1$ which accepts $L_1 = \{ a^n b^n c^n \, | \, n \ge 0 \}$.}
\label{diagram}
\label{ra-anbncn}
\end{figure}

\begin{figure}[t]
\centerline{
\includegraphics[scale=0.65]{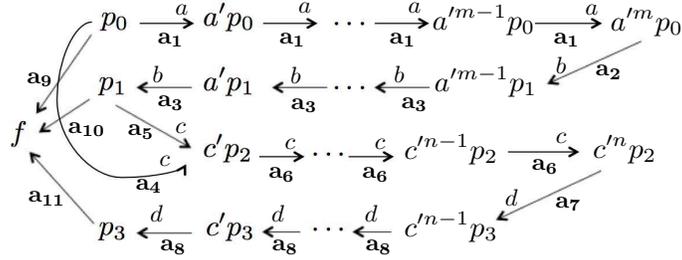}}
\caption{Reaction diagram of $\mathcal{A}_2$ which accepts $L_2 = \{ a^m b^m c^n d^n \, | \, m, n \ge 0 \}$.}
\label{diagram}
\label{ra-anbncmdm}
\end{figure}

\begin{lem}
For an alphabet $\Sigma$ with $|\Sigma| \ge 2$, let $h:\Sigma^* \rightarrow \Sigma^*$ be an injection such that for any $w \in \Sigma^*$, $|h(w)|$ is bounded by a polynomial of $|w|$. Then, there is no polynomially-bounded reaction automaton which accepts the language $L = \{ wh(w) \, | \,  w \in \Sigma^* \}$. \label{lem-wfw}
\end{lem}

\begin{proof}
Assume that there is a {\it poly}-RA $\mathcal{A} = (S, \Sigma, A, D_0, f)$ such that $L(\mathcal{A}) = \{ wh(w) \, | \,  w \in \Sigma^* \}$.  Let $|S| = m_1$, $|\Sigma|= m_2 \ge 2$ and the input string be $wh(w)$ with $|w|=n$. 

Since $|h(w)|$ is bounded by a polynomial of $|w|$, $|wh(w)|$ is also bounded by a polynomial of $n$.  Hence, for each $D_i$ in an interactive process $\pi \in IP( \mathcal{A}, wh(w) )$, it holds that $|D_i| \le p(n)$ for some polynomial $p(n)$ from the definition of a {\it poly}-RA.

Let $\mathcal{D}_{p(n)} = \{ D \in S^\# \, | \, |D| \le p(n) \}$. Then, it holds that
\begin{align*}
&|\mathcal{D}_{p(n)}| = \sum^{p(n)}_{k = 0} {}_{m_1} \mathrm{H}_k = \sum^{p(n)}_{k = 0} \frac{(k + m_1 -1)!}{k! \cdot (m_1 -1) !} = \frac{(p(n) + m_1)!}{p(n)! \cdot m_1 !} = \frac{(p(n) + m_1)(p(n) + m_1 - 1) \cdots (p(n) + 1)}{ m_1 !}.\\
&({}_{m_1} \mathrm{H}_k \text{ denotes the number of repeated combinations of $m_1$ things taken $k$ at a time.})
\end{align*}
Therefore, there is a polynomial $p'(n)$ such that $|\mathcal{D}_{p(n)}| = p'(n)$. Since it holds that $|\Sigma^n| = (m_2)^n$, if $n$ is sufficiently large,  we obtain the inequality $|\mathcal{D}_{p(n)}| < |\Sigma^n|$. 

For $i \ge 0$ and $w \in \Sigma^*$, let $I_i(w) = \{ D_i \in \mathcal{D}_{p(n)} \, | \, \pi = D_0, \ldots, D_i, \ldots \in IP(\mathcal{A}, w) \} \subseteq \mathcal{D}_{p(n)}$, i.e., $I_i(w)$ is the set of multisets in $\mathcal{D}_{p(n)}$ which appear as the $i$-th elements of interactive processes in $ IP(\mathcal{A}, w)$. From the fact that $L(\mathcal{A}) = \{ wh(w) \, | \,  w \in \Sigma^* \}$ and $h$ is an injection, we can show that for any two distinct strings $w_1, w_2 \in \Sigma^n$, $I_n(w_1)$ and $I_n(w_2)$ are incomparable. This is because if $I_n(w_1) \subseteq I_n(w_2)$, the string $w_2 h(w_1)$ is accepted by $\mathcal{A}$, which means that $h(w_1) = h(w_2)$ and contradicts that $h$ is an injection. 

Since for any two distinct strings $w_1, w_2 \in \Sigma^n$, $I_n(w_1)$ and $I_n(w_2)$ are incomparable and $I_n(w_1), I_n(w_2) \subseteq \mathcal{D}_{p(n)}$, it holds that 
\[ | \{ I_n(w) \, | \, w \in \Sigma^n \} | \le |\mathcal{D}_{p(n)}| < |\Sigma^n|. \]
However, from the pigeonhole principle, the inequality $| \{ I_n(w) \, | \, w \in \Sigma^n \} | < |\Sigma^n|$ contradicts that for any two distinct strings $w_1, w_2 \in \Sigma^n$, $I_n(w_1) \ne I_n(w_2)$.
\end{proof}

\begin{thm} The following inclusions hold {\rm :}  \\
{\rm (1)}. $\mathcal{REG}=k$-$\mathcal{RA} 
\subset  \mathcal{LRA} \subseteq \mathcal{PRA} \subset \mathcal{ERA} \subseteq \mathcal{RA} =  \mathcal{RE}$ {\rm (for each $k\geq 1$)}. \\
{\rm (2)}. $\mathcal{LRA} \subset \mathcal{CS} \subseteq \mathcal{ERA}$. \\
{\rm (3)}. $\mathcal{LIN}$ {\rm (}$\mathcal{CF}${\rm )} and $\mathcal{LRA}$ are incomparable.
\end{thm}
\begin{proof}
(1). From the definitions, the inclusion $\mathcal{REG}\subseteq 1$-$\mathcal{RA}$ is straightforward. Conversely, for a given $k$-RA $\mathcal{A}=(S, \Sigma, A, D_0, f)$ and for  $w \in L(\mathcal{A})$, there exists a $\pi$ in $IP(\mathcal{A},w)$ such that for each $D_i$ appearing in $\pi$, we have $|D_i| \le k$. Let $Q=\{ D \in S^{\#} \mid |D|\leq k\}$ and $F=\{ D \mid D\in Q, f\subseteq D, Res_A (D) = \{ D \} \}$, and  
 construct an NFA $M=(Q, \Sigma, \delta, D_0, F)$, where $\delta$ is defined by $\delta(D,a)\ni D'$ if $D\rightarrow^a D'$ for $a \in 
\Sigma\cup \{\lambda\}$. Then, it is seen that $L(\mathcal{A})=L(M)$, and $k$-$\mathcal{RA}\subseteq \mathcal{REG}$, thus we obtain that $\mathcal{REG}= k$-$\mathcal{RA}$.  The other inclusions are all obvious 
from the definitions. The language $L=\{a^nb^n\mid n\geq 0\}$ proves 
the proper inclusion : $\mathcal{REG}\subset \mathcal{LRA}$. A proper inclusion $\mathcal{PRA} \subset \mathcal{ERA}$ is due to that $L_3 =\{ ww^R \mid w\in \{a,b\}^* \} \in \mathcal{ERA} - \mathcal{PRA}$,  
which follows from Lemma \ref{lem-wfw}. \\
(2). Given an $lin$-RA $\mathcal{A}$, one can consider a linearly bounded automaton (LBA) $M$ that simulates an interactive process $\pi$ in $IP(\mathcal{A},w)$ for each $w$, because of the linear boundedness of $\mathcal{A}$. This implies that $\mathcal{LRA} \subseteq \mathcal{CS}$. 
A proper inclusion is due to that $L_3 =\{ ww^R \mid w\in \{a,b\}^* \} \in \mathcal{LIN}-\mathcal{LRA}$,  
which follows from Lemma \ref{lem-wfw}.

Further, for a given LBA $M$, one can find an equivalent two-stack machine $M_s$ whose stack lengths are linearly bounded by the input length. This implies, from the proof of Theorem \ref{teiri1},  that $M_s$ is simulated by an RA $\mathcal{A}$ that is exponentially bounded. Thus, it holds that $\mathcal{CS} \subseteq \mathcal{ERA}$. \\
(3). The language $L_1=\{a^nb^nc^n\mid n\geq 0\}$ (resp. 
$L_2=\{a^m b^m c^n d^n \mid m,n \geq 0\}$)  is in 
$\mathcal{LRA} - \mathcal{CF}$ (resp. \ $\mathcal{LRA} -  \mathcal{LIN})$, while, again from Lemma \ref{lem-wfw},  
the language $L_3$ is in 
$\mathcal{LIN}-\mathcal{LRA}$. This completes the proof. 
\end{proof}

\begin{figure}[t]
\centerline{
\includegraphics[scale=0.65]{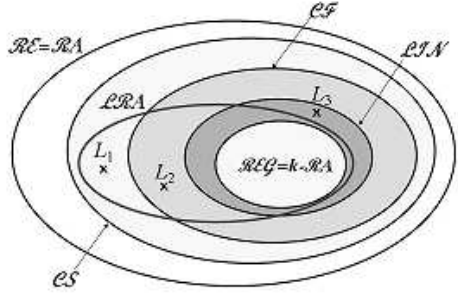}}
\caption{Language class relations in the Chomsky hierarchy : $L_1=\{a^nb^nc^n \mid n \geq 0\}$; \ $L_2=\{a^m b^m c^n d^n \mid m, n \geq 0\}$;\ 
$L_3=\{w w^R \mid w \in \{a,b\}^* \}$ }
\label{hie}
\end{figure}

\section{Concluding Remarks}

Based on the formal framework presented in a series of papers \cite{ER:07a,ER:07b,ER:09,EMR:10,EMR:11}, we have introduced the notion of reaction automata and investigated the language accepting powers of the automata. Roughly, a reaction automaton may  be characterized in terms of  three key words as follows : a {\it language accepting device} based on the {\it multiset rewriting} in the {\it maximally parallel manner}. Specifically, we have shown that in  a  computing schema with one-pot solution and a finite number of molecular species, reaction automata can perform the Turing universal computation. 
The idea behind their computing principle is to simulate the behavior of two pushdown stacks in terms of  multiset rewriting with the help of an encoding technique, where both the manner of maximally parallel rewriting and the role of the inhibitors in each reaction are effectively utilized.   

There already exist quite a few number of literature investigating on the notion of a multiset and its related topics (\cite{CPRS:01}) in which 
multiset automata and grammars are formulated and explored largely from the formal language theoretic point of view. Rather recent papers (\cite{KTZ:09a,KTZ:09b}) focus on the accepting power of multiset pushdown automata to characterize the classes of multiset languages through  investigating  their closure properties.    

To the authors' knowledge, however, relatively few works have devoted to computing languages with multiset rewriting/communicating mechanism. Among them, one can find some papers published in the area of membrane computing (or spiking neural P-systems)  where a string is encoded in some manner as a natural number and a language is specified  as a set of natural numbers (e.g., \cite{CIPP:06}). Further, recent developments concerning P-automata and its variant called dP-automata are noteworthy in the sense that they may give rise to a new type of computing devices that could be a bridge between P-system theory and the theory of reaction systems and automata (\cite{COV:10,IPPY:11,PP:11}). 

In fact, a certain number of computing devices similar to reaction automata have already been investigated in the literature.  Among others, parallel labelled rewrite transition systems are proposed and investigated (\cite{HM:01}) in which 
multiset automata may be regarded as special type of reaction automata, whereas neither regulation by inhibitors nor  maximally parallel manner of applying rules is employed in their rewriting process.  A quite recent   
article \cite{AV:11} investigates the power of maximally parallel multiset rewriting systems (MPMRSs) and proves the existence of a universal MPMRS having smaller number of rules, which directly implies the existence of a universal antiport P-systems, with one membrane, having smaller number of rules. In contrast to reaction automata, a universal MPMRS computes any partially recursive function provided that the input is the encoding of a register machine computing a target function. 

Turning to the formal grammars, one can find  random context grammars (\cite{DP:01}) and their variants (such as semi-conditional grammars in \cite{Paun:85}) that employ regulated rewriting mechanisms called permitting symbols and forbidding symbols. The roles of these two are corresponding to reactants and inhibitors in reactions, whereas they deal with sets of strings (i.e., languages in the usual sense) rather than multisets. 
We finally refer to an article on stochastic computing models based on chemical kinetics,  which  proves that  well-mixed 
finite stochastic chemical reaction networks with a fixed number of species can achieve Turing universal computability with an arbitrarily low error probability (\cite{SCWB:08}). In this paper, we have shown that non-stochastic chemical  reaction systems with finite number of molecular species  can also achieve Turing universality with the help of 
 inhibition mechanism.

Many subjects remain to be investigated along the research direction suggested by reaction automata in this paper. First, it is of importance to completely characterize the computing powers and the closure properties of complexity subclasses of reaction automata introduced in this paper.   
Secondly, from the viewpoint of designing chemical reactions, it is useful to explore a methodology for ``chemical reaction programming'' in terms of  reaction automata. It is also interesting to simulate a variety of chemical reactions in the real world by the use of the framework of reaction automata. 

\section*{Acknowledgements}

The authors gratefully acknowledge useful remarks and comments by anonymous referees which improved an earlier version of this paper.
The work of F. Okubo was possible due to Waseda University Grant for Special Research Projects: 2011A-842.
The work of S. Kobayashi was in part supported by Grants-in-Aid for Scientific Research (C) No.22500010, Japan Society for the Promotion of Science.
The work of T. Yokomori was in part supported by Waseda University Grant for Special Research Projects: 2011B-056.

\end{document}